\def\year{2022}\relax
\documentclass[letterpaper]{article} 

\usepackage{aaai22}  
\usepackage{times}  
\usepackage{helvet}  
\usepackage{courier}  
\usepackage[hyphens]{url}  
\usepackage{graphicx} 
\usepackage{natbib}  
\usepackage{caption} 
\DeclareCaptionStyle{ruled}{labelfont=normalfont,labelsep=colon,strut=off} 
\usepackage{hyperref} 
\usepackage{ellipsis, ragged2e} 

\usepackage[english]{babel}
\usepackage[utf8]{inputenc}
\usepackage[final]{microtype}

\usepackage{amsmath,amssymb,amsthm}
\usepackage{csquotes}
\usepackage{booktabs}

\usepackage{pgfplots}
\pgfplotsset{compat=1.16}
\usepackage{tikz}
\usetikzlibrary{shapes,positioning,arrows,arrows.meta,calc,automata,matrix,fit,backgrounds}
\tikzstyle{state}+=[minimum size = 6mm, inner sep=0,outer sep=1]

\colorlet{disabled}{lightgray}
\tikzstyle{state}=[draw,rectangle,inner sep=5pt,rounded corners=2pt]
\tikzstyle{action}=[font=\small,inner sep=0pt,outer sep=3pt]
\tikzstyle{actionnode}=[circle,draw=black,fill=black,minimum size=1mm,inner sep=0,outer sep=0]
\tikzstyle{actionedge}=[draw,-]
\tikzstyle{prob}=[font=\scriptsize,inner sep=0pt,outer sep=1pt]
\tikzstyle{probedge}=[draw,->]
\tikzstyle{directedge}=[draw,->]
\tikzset{chainarrow/.tip={Stealth[length=3pt]}}
\tikzset{>=chainarrow}

\usepackage{xparse}
\usepackage{mathtools}
\usepackage{environ}
\usepackage{marvosym}
\usepackage{bbm}
\usepackage{fontawesome5}

\newtheorem{theorem}{Theorem}
\newtheorem{corollary}{Corollary}
\newtheorem{lemma}{Lemma}
\theoremstyle{definition}
\newtheorem{definition}{Definition}
\newtheorem{remark}{Remark}
\newtheorem{example}{Example}

\usepackage{algorithmicx,algorithm}
\usepackage[noend]{algpseudocode}

\usepackage[capitalize]{cleveref}
\DeclarePairedDelimiter{\delimabs}{\lvert}{\rvert}
\DeclarePairedDelimiter{\delimcardinality}{\lvert}{\rvert}
\DeclarePairedDelimiter{\delimnorm}{\lVert}{\rVert}

\NewDocumentCommand{\abs}{sm}{\IfBooleanTF{#1}{\delimabs*{#2}}{\delimabs{#2}}}
\NewDocumentCommand{\cardinality}{sm}{\IfBooleanTF{#1}{\delimcardinality*{#2}}{\delimcardinality{#2}}}
\NewDocumentCommand{\norm}{sm}{\IfBooleanTF{#1}{\delimnorm*{#2}}{\delimnorm{#2}}}
\NewDocumentCommand{\powerset}{r()}{2^{#1}}

\newcommand{\unionSym}{\cup}
\newcommand{\unionBin}{\mathbin{\unionSym}}

\newcommand{\UnionSym}{\bigcup}

\newcommand{\union}{\unionBin}

\newcommand{\Union}{\UnionSym}

\newcommand{\Naturals}{\mathbb{N}}

\newcommand{\Reals}{\mathbb{R}}

\NewDocumentCommand{\Measures}{d()}{\IfNoValueTF{#1}{\Pi}{\Pi(#1)}}
\NewDocumentCommand{\integral}{d<> m m}{\IfNoValueTF{#1}{\int #2\,d#3}{\int_{#1} #2\,d#3}}
\NewDocumentCommand{\Expectation}{s d[]}{\IfNoValueTF{#2}{\mathbb{E}}{\mathbb{E}\IfBooleanTF{#1}{\left[#2\right]}{[#2]}}}
\NewDocumentCommand{\Probability}{s d[]}{\mathop{\mathrm{Pr}}\IfValueT{#2}{\IfBooleanTF{#1}{\left[#2\right]}{[#2]}}}

\newcommand{\MC}{\mathsf{M}}
\newcommand{\MDP}{\mathcal{M}}

\newcommand{\States}{S}
\newcommand{\initialstate}{\hat{s}}
\newcommand{\Actions}{A}
\NewDocumentCommand{\stateactions}{r()}{{\Actions}(#1)}
\NewDocumentCommand{\mctransitions}{d()}{\IfNoValueTF{#1}{\delta}{\delta(#1)}}
\NewDocumentCommand{\mdptransitions}{d()}{\IfNoValueTF{#1}{\Delta}{\Delta(#1)}}
\NewDocumentCommand{\actionstate}{r<> r()}{s_{#1}(#2)}
\NewDocumentCommand{\actioncost}{d()}{\IfNoValueTF{#1}{{C}}{{C}(#1)}}

\NewDocumentCommand{\Infinitepaths}{d<>}{\IfNoValueTF{#1}{\mathsf{Paths}}{\mathsf{Paths}_{#1}}}
\NewDocumentCommand{\Finitepaths}{d<>}{\IfNoValueTF{#1}{\mathsf{FPaths}}{\mathsf{FPaths}_{#1}}}
\newcommand{\strategy}{\pi}
\NewDocumentCommand{\Strategies}{d<>}{\IfNoValueTF{#1}{\Pi}{\Pi_{#1}}}
\NewDocumentCommand{\StrategiesM}{d<>}{\IfNoValueTF{#1}{\Pi}{\Pi_{#1}}^{\mathsf{M}}}
\NewDocumentCommand{\StrategiesMD}{d<>}{\IfNoValueTF{#1}{\Pi}{\Pi_{#1}}^{\mathsf{MD}}}

\DeclareMathOperator{\SccsOp}{SCC}
\DeclareMathOperator{\BsccsOp}{BSCC}
\DeclareMathOperator{\EcsOp}{EC}
\DeclareMathOperator{\MecsOp}{MEC}

\NewDocumentCommand{\Sccs}{d()}{\IfNoValueTF{#1}{\SccsOp}{\SccsOp(#1)}}
\NewDocumentCommand{\Bsccs}{d()}{\IfNoValueTF{#1}{\BsccsOp}{\BsccsOp(#1)}}
\NewDocumentCommand{\Ecs}{d()}{\IfNoValueTF{#1}{\EcsOp}{\EcsOp(#1)}}
\NewDocumentCommand{\Mecs}{d()}{\IfNoValueTF{#1}{\MecsOp}{\MecsOp(#1)}}

\NewDocumentCommand{\ProbabilityMC}{s r<> d[]}{\mathsf{Pr}_{#2}\IfNoValueF{#3}{\IfBooleanTF{#1}{\!\left[#3\right]\!}{[#3]}}}
\NewDocumentCommand{\ProbabilityMDP}{s r<> r<> d[]}{\mathsf{Pr}_{#2}^{#3}\IfNoValueF{#4}{\IfBooleanTF{#1}{\!\left[#4\right]\!}{[#4]}}}
\NewDocumentCommand{\ProbabilityMDPmax}{s r<> d[]}{\mathsf{Pr}_{#2}^{\max}\IfNoValueF{#3}{\IfBooleanTF{#1}{\!\left[#3\right]\!}{[#3]}}}
\NewDocumentCommand{\ProbabilityMDPsup}{s r<> d[]}{\mathsf{Pr}_{#2}^{\sup}\IfNoValueF{#3}{\IfBooleanTF{#1}{\!\left[#3\right]\!}{[#3]}}}

\NewDocumentCommand{\ExpectationMDP}{s r<> r[]}{\mathbb{E}\IfBooleanTF{#1}{\!\left[#3 \mid #2 \right]\!}{[#3 \mid #2]}}

\newcommand{\goalset}{\mathcal{G}}
\NewDocumentCommand{\expectedval}{r<> d()}{V^{#1}\IfValueT{#1}{(#2)}}
\NewDocumentCommand{\totalsum}{d<>}{\IfValueTF{#1}{\mathcal{R}^{#1}}{\mathcal{R}}}

\NewDocumentCommand{\nstepprob}{r<> d<> d()}{\IfValueTF{#2}{p_{#1}^{#2}}{p_{#1}}\IfValueT{#3}{(#3)}}
\NewDocumentCommand{\nstepnongoal}{r<> d<>}{\IfValueTF{#2}{\mathcal{N}_{#1}^{#2}}{\mathcal{N}_{#1}}}
\NewDocumentCommand{\expectedcost}{d<> d()}{\IfValueTF{#1}{e^{#1}}{e}\IfValueT{#2}{(#2)}}
\NewDocumentCommand{\nstepexpectedcost}{r<> d<>}{\IfValueTF{#2}{\mathcal{E}_{#1}^{#2}}{\mathcal{E}_{#1}}}
\NewDocumentCommand{\pareto}{r<> r<>}{\mathfrak{P}_{#1}^{#2}}
\NewDocumentCommand{\Distributions}{r[]}{\mathcal{D}[#1]}

\DeclareMathOperator{\CVaR}{CVaR}
\DeclareMathOperator{\VaR}{VaR}

\newcommand{\threshold}{\mathsf{t}}

\DeclareMathOperator{\conv}{conv}

\title{Risk-aware Stochastic Shortest Path}
\author{
	Tobias Meggendorfer
}
\affiliations{
	IST Austria\\
	Am Campus 1, 3400 Klosterneuburg, Austria\\
	\texttt{tobias.meggendorfer@ist.ac.at}
}

\usepackage{todonotes}

\advance\textwidth6mm
\advance\hoffset-3mm
\begin{document}

\maketitle

\begin{abstract}
	We treat the problem of \emph{risk-aware control} for \emph{stochastic shortest path} (SSP) on \emph{Markov decision processes} (MDP).
	Typically, expectation is considered for SSP, which however is oblivious to the incurred risk.
	We present an alternative view, instead optimizing \emph{conditional value-at-risk} (CVaR), an established risk measure.
	We treat both Markov chains as well as MDP and introduce, through novel insights, two algorithms, based on linear programming and value iteration, respectively.
	Both algorithms offer precise and provably correct solutions.
	Evaluation of our prototype implementation shows that risk-aware control is feasible on several moderately sized models.
\end{abstract}

\section{Introduction}

\emph{Markov decision processes} (MDP) are a standard model for sequential decision making in uncertain environments, applied in, for example, robot motion planning; see e.g.\ \cite{white1993survey,white1985real} for a variety of further examples.
Usually, one aims to control such a system optimally with respect to a \enquote{performance rating}, called \emph{objective}.
In this work, we consider the \emph{stochastic shortest path} (SSP) objective \cite{DBLP:journals/mor/BertsekasT91}, where the goal is to minimize the accumulated cost until a given set of target states is reached.

Traditionally, one seeks a policy minimizing the expectation of this accumulated cost.
However, this policy willingly accepts arbitrary risks to achieve a minimal increase in expected profit.
This may be undesirable, especially when the system in question, for example, models a situation which takes a long time to unfold or is only executed once, such as a retirement savings plan or Mars rover path planning.
In particular, the law of large numbers is not applicable, and expectation alone provides little insight in the actual dynamics.

To remedy this issue, \emph{risk-aware control} proposes several ideas.
A popular approach is to quantify the risk incurred by a policy and then optimizing this \emph{risk measure} instead of expectation.
We briefly discuss some relevant measures:
\emph{Variance} does not focus on bad cases and may even incentivize intentionally performing suboptimally in unexpectedly good situations.
\emph{Worst case} analysis often is too pessimistic in probabilistic environments, considering events with probability $0$ such as a fair coin toss never yielding heads.
\emph{Value-at-risk} (VaR) is the worst $p$-quantile for a given threshold $\threshold \in [0, 1]$.
It approximates the notion of a \enquote{reasonably likely} bad case.
However, VaR ignores the magnitude of worse cases, and has been characterized \enquote{seductive, but dangerous} and \enquote{not sufficient to control risk} \cite{beder1995var}.
\emph{Conditional value-at-risk} (CVaR) (average value-at-risk, expected shortfall) yields the expectation over all outcomes worse than the VaR, i.e.\ the \enquote{tail loss}; see \cref{fig:cvar_example} for a sketch.
As such, it considers outliers, weighted accordingly.
It is an established and \enquote{more consistent measure of risk} \cite{artzner1999coherent,rockafellar2000optimization}, gaining popularity in various fields.
We direct the interested reader to \cite{sarykalin2008value} for detailed comparison between VaR and CVaR, \cite{DBLP:journals/itor/FilippiGS20} for a survey of CVaR applications, and \cite{DBLP:books/daglib/0034641} for further risk measures.

Motivated by these observations, our primary goal in this work is to provide risk-aware control for SSP objectives through the optimization of its CVaR.

\begin{figure}[t]
	\centering
	\begin{tikzpicture}
		\begin{axis}[width=0.9\columnwidth,height=3cm,
			xmin=1, xmax=11, ymin=0, ymax=0.55, ybar,
			bar shift=0pt, axis x line*=middle,
			ytick={0.1,0.3,0.5},
			x label style={at={(axis description cs:-0.05,-0.1)},anchor=north,font=\small},
			y label style={font=\small},
			xlabel={Costs},
			ylabel={Probability}
		]
			\addplot [draw=black,fill=white] coordinates {(2,0.2) (5,0.35) (7,0.25) (8,0.05) (9,0.15)};
			\addplot [fill=gray] coordinates {(7,0.2) (8,0.05) (9,0.15)};
			\node[anchor=north east,inner sep=0pt,outer sep=2pt] (var) at (7,0.5) {$\VaR_{40\%}$};
			\draw[->] (7,0.5) -- (7,0.0);
			\node[anchor=north west,inner sep=0pt,outer sep=2pt] (cvar) at (7.875,0.5) {$\CVaR_{40\%}$};
			\draw[->] (7.875,0.5) -- (7.875,0.0);
		\end{axis}
	\end{tikzpicture}
	\caption{
		Example distribution over costs to showcase VaR and CVaR with threshold $\threshold = 40\%$.
		The bars represent the respective probabilities, while the grey area depicts the part considered by CVaR:
		The sum of the grey area equals the specified threshold of $40\%$, the expectation over it is $7.875$.
	} \label{fig:cvar_example}
\end{figure}
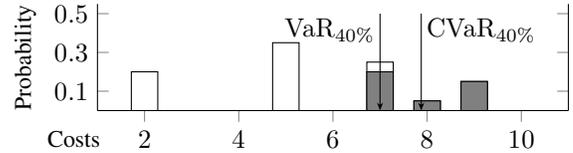

\paragraph{Related Work}
Firstly, \cite{DBLP:conf/lics/KretinskyM18} also considers CVaR, however for mean payoff instead of SSP and only provides an linear programming based solution, while we additionally present a value iteration approach.
Secondly, \cite{DBLP:conf/nips/ChowTMP15} treats the \emph{discounted} variant of the problem, fundamentally relying on the discounting factor to bound the error of the approximation.
Moreover, \cite{DBLP:conf/icra/CarpinCP16} considers two variants of the problem.
For the first, much more restricted variant, they present an approximative formulation together with an algorithm.
For the second variant, which is more general and closer to our problem, they present an approximation of the previous approximative formulation, without any guaranteed bounds.
In contrast to the former two, our approach yields a provably correct, precise result for the infinite horizon problem on MDP, only using standard assumptions.

Many other works dealing with CVaR on MDP, e.g.\ \cite{DBLP:journals/tac/BorkarJ14,DBLP:conf/aaai/KeramatiDTB20,DBLP:journals/mmor/BauerleO11}, often consider finite horizon and/or discounted costs, but not the undiscounted infinite horizon variant.
A broad spectrum of research focusses on risk-aware \emph{reinforcement learning}, typically aiming at best-effort solutions converging to the optimal value in the limit at most and without any guarantees.
Note that especially when considering risk, providing reliable guarantees may be considered vital.
Often, these best-effort solutions also introduce additional constraints, such as restricting to a suboptimal class of policies. 
A different perspective considers \emph{time consistent} risk measures \cite{DBLP:journals/mp/Ruszczynski10}, where risk effectively \enquote{accumulates} along the run; see \cite{DBLP:journals/tac/TamarCGM17} for a CVaR-variant.

\paragraph{Contributions \& Novelty}
We treat, to our knowledge for the first time, the problem of optimizing the \emph{global}, \emph{infinite horizon} risk of SSP on MDP through CVaR, providing \emph{provably correct results}.
We first discuss the problem on Markov chains, derive a central \emph{characterization of CVaR for SSP}, and provide a tailored algorithm.
Then, we present \emph{two novel solution approaches} for MDP, based on this characterization.
One is based on linear programming and one on value iteration (both exponential).
While the primary focus of this work is the theoretical contribution, we also evaluate a prototype implementation, showing that risk-aware control with guarantees is practical even for infinite horizon problems.

\section{Preliminaries}

A (finite, discrete time, time-homogeneous) \emph{Markov decision process} (MDP) \cite{DBLP:books/wi/Puterman94,DBLP:books/lib/Bertsekas05} is a tuple $\MDP = (\States, \Actions, \initialstate, \mdptransitions, \actioncost)$, where
	$\States$ is a finite set of \emph{states},
	$\Actions$ is a finite set of \emph{actions},
	$\initialstate \in \States$ is the \emph{initial state},
	$\mdptransitions(s, a, s') = \Probability[s' | s, a]$ is the Markovian \emph{transition function}, and
	$\actioncost(s, a) \in \Naturals_0$ is the non-negative, integer \emph{cost} associated with taking action $a$ in state $s$.
We choose integer costs for simplicity, however our methods are also applicable to rational costs (by rescaling).
An action $a$ is \emph{available} in state $s$ if $\sum_{s' \in \States} \mdptransitions(s, a, s') = 1$ (the sum is $0$ otherwise).
We write $\Actions(s) \subseteq \Actions$ for the set of all actions available in state $s$.

A \emph{Markov chain} (MC) is an MDP where $\cardinality{\Actions(s)} = 1$ for all states $s \in \States$, i.e.\ the system is fully probabilistic.

The non-determinism in MDP is resolved by \emph{policies}, mappings from finite paths to distributions over actions.
The set of all policies is denoted by $\Strategies$.
A policy is called
	(i)~\emph{deterministic} if it always yields a unique action,
	(ii)~\emph{memoryless} (or \emph{stationary}) if it only depends on the current state, and
	(iii)~\emph{Markovian} if it only depends on the number of steps already performed.
Technically, an MDP with a policy \emph{induces} a Markov chain, which allows to reason about the now fully probabilistic system.
See, e.g., \cite[Chp.~2]{DBLP:books/wi/Puterman94} or \cite[Sec.~10.6]{DBLP:books/daglib/0020348} for formal details.

\paragraph{Stochastic Shortest Path (SSP)} \cite{DBLP:journals/mor/BertsekasT91} is a common objective on MDP, specified by an MDP and a set of \emph{goal states} $\goalset \subseteq \States$.
We are interested in the total accumulated cost until a goal state is reached.
We write $\totalsum<s, \strategy>$ to denote the distribution over total costs achieved by policy $\strategy$ starting in state $s$.
Typically, one optimizes the \emph{expected} total cost, i.e.\ given a state $s$, find a policy $\strategy \in \Strategies$ such that
\begin{equation*}
	\expectedval<\strategy>(s) \coloneqq \Expectation*[{\sum}_{t=0}^\infty \actioncost(s_t, a_t) \mid s, \strategy] = \Expectation[\totalsum<s, \strategy>].
\end{equation*}
is minimal.
As already suggested, instead of expectation, we however are interested in optimizing a \emph{risk measure} of $\totalsum<s, \strategy>$.

\paragraph{Conditional Value-at-Risk (CVaR)} (also known as Average Value-at-Risk (AVaR)) is our proposed alternative to expectation.
To introduce CVaR, we first need to define the notion of \emph{value-at-risk} (VaR).
Intuitively, VaR tries to answer the question \enquote{what is a reasonable bad outcome?}
VaR is parametrized by a threshold $\threshold \in [0, 1]$ and yields the \emph{worst $\threshold$-quantile}, i.e.\ a value $v$ such that an outcome is worse than $v$ with probability $\threshold$.
For example, the $50\%$-VaR effectively is the median.
Formally, given a distribution over natural numbers $X : \Naturals_0 \to [0, 1]$ and $\threshold < 1$ we define
\begin{equation*}
	\VaR_\threshold(X) \coloneqq \min \{v \in \Naturals_0 \mid {\sum}_{x=v + 1}^\infty X(x) \leq t\}.
\end{equation*}
(As we are considering costs, larger values are worse.)
For consistency, let $\VaR_1(X) \coloneqq \min\{v \in \Naturals_0 \mid X(v) > 0\}$.
Note that for $\threshold = 0$ we may have $\VaR_0(X) = \infty$.
\begin{example} \label{ex:var_fig}
	Consider the distribution from \cref{fig:cvar_example}, i.e.\ $X = \{2 \mapsto 20\%, 5 \mapsto 35\%, 7 \mapsto 25\%, 8 \mapsto 5\%, 9 \mapsto 15\%\}$.
	We see that $\VaR_{40\%}(X) = 7$.
	Maybe unexpectedly, we have $\VaR_{45\%}(X) = 5$ instead of $7$, as $\Probability[X > 5]$ is exactly $45\%$.
	It is a matter of preference how to define this boundary case, and either works in our setting.
	In particular, it does not influence the definition of CVaR.
\end{example}
CVaR, also parametrized by a threshold $\threshold \in [0, 1]$, aims to answers the question \enquote{what can we expect from an average bad case?}
Formally, CVaR equals the expectation of $X$ conditional on only considering the worst $\threshold$ outcomes.
Similar to \cref{ex:var_fig}, we need to apply special care when working with discrete distributions:
Again recall the example from \cref{fig:cvar_example} with $\threshold = 40\%$.
There, only $20\%$ of the $X(7) = 25\%$ should be considered.
Thus, CVaR is defined as follow.
For a distribution $X$ and threshold $\threshold > 0$, define $v \coloneqq \VaR_\threshold(X)$ and $\mathfrak{V} \coloneqq X > v$ the event of an outcome being strictly worse than the VaR.
Then
\begin{equation}
	\CVaR_\threshold(X) \coloneqq \tfrac{1}{\threshold} \big(\Probability[\mathfrak{V}] \cdot \Expectation[X \mid \mathfrak{V}] + \\ (\threshold - \Probability[\mathfrak{V}]) \cdot v \big). \label{eq:cvar_definition}
\end{equation}
For the degenerate case of $\threshold = 0$, we define $\CVaR_0(X) \coloneqq \lim_{\threshold \to 0} \CVaR_\threshold(X) = \VaR_0(X)$.
\begin{remark}
	Observe that $\CVaR_0(X)$ is the worst-case of $X$ and $\CVaR_1(X) = \Expectation[X]$ the expectation of $X$; changing $\threshold$ thus \emph{smoothly interpolates} between these extremes.
	As both of these extremal cases are already solved for SSP, we exclude them, i.e.\ assume $0 < \threshold < 1$.
\end{remark}
See e.g.\ \cite[Sec.~3]{DBLP:conf/lics/KretinskyM18} for a more detailed discussion of CVaR on discrete distributions.

\paragraph{Problem Statement}
Together, the central question considered in this work is:
\begin{quote}
	\emph{Given an SSP problem, what is the optimal CVaR?}
\end{quote}
Formally, given an MDP $\MDP$, cost function $\actioncost$, and threshold $0 < \threshold < 1$, we want to determine
\begin{equation*}
	\CVaR_\threshold^* \coloneqq {\inf}_{\strategy \in \Strategies} \CVaR_\threshold(\totalsum<\initialstate, \strategy>).
\end{equation*}
We refer to this problem as \textbf{CVaR-SSP}.
%

\paragraph{Linear Programming (LP)} (see e.g.\ \cite{DBLP:books/daglib/0090562}) is an established problem solving technique with strong connections to MDP; many popular objectives allow for a natural LP formulation.
An LP is characterized by a linear \emph{objective function} $f$ and a set of \emph{linear inequality constraints} on its variables. 
The task then is to find the maximal (or minimal) value of $f$ subject to the imposed constraints. 
This value can be computed in polynomial time \cite{khachiyan1979polynomial,DBLP:journals/combinatorica/Karmarkar84}.
As such, LP is a popular tool to prove complexity bounds of many problems.

\paragraph{Value Iteration (VI)} \cite{bellman1966dynamic} is a popular practical approach to solve various questions related to MC and MDP, among others.
As the name suggests, one repeatedly applies an iteration operator to a value vector $v_i$ (typically one real value per state).
For example, the canonical value iteration for SSP starts with $v_0(s) = 0$ for all $s \in \States$ and then iterates
\begin{equation*}
	v_{i+1}(s) = {\min}_{a \in \Actions(s)} \actioncost(s, a) + {\sum}_{s' \in \States} \mdptransitions(s, a, s') \cdot v_i(s').
\end{equation*}
Under the mentioned assumptions, this iteration converges to the true value in the limit, with an exponential worst-case bound to reach a given precision.
In practice, VI typically performs very well, quite often outperforming LP approaches by a large margin.
A similar trend emerges for our approaches.

\paragraph{Assumptions}
Finally, we introduce several standard assumptions for \textbf{CVaR-SSP}.
First, we assume that $\initialstate \notin \goalset$ (otherwise the problem would be trivial) and that all goal states are absorbing.
Next follow two standard assumptions for SSP \cite{DBLP:books/lib/BertsekasT96}.
A policy is \emph{proper} if the probability of eventually reaching the goal from every state is $1$.
We assume that
	(i)~there exists a proper policy and
	(ii)~for every improper policy $\strategy$, $\expectedval<\strategy>(s)$ is infinite for at least one state $s$.
Finally, we assume that the cost of an action is 0 \emph{if and only if} the corresponding state is a goal state.
This assumption, also used in e.g.\ \cite{DBLP:journals/mor/Bonet07,DBLP:conf/icra/CarpinCP16}, is mainly introduced for simplicity, we briefly discuss later on how it can be lifted.
Note that (ii) follows from (i) and the latter assumption.
\section{Reachability \& Uniform Costs}

We restrict to a simpler setting to explain central insights more clearly.
Namely, we assume that costs are uniform, i.e.\ $\actioncost(s, a) = 1$ for all non-goal states.
The total cost $\totalsum<s, \strategy>$ now can be interpreted as \enquote{starting from state $s$ with policy $\strategy$, how many steps are needed to reach the goal?}
The VaR corresponds to the first step after which a fraction of at least $1 - \threshold$ of all executions (abbreviated by \enquote{$1 - \threshold$ executions} in the following) have reached a goal state; CVaR is the overall expected number of steps to reach the goal for the remaining $\threshold$ executions.
We discuss the general case afterwards.

\subsection{Markov Chains}

To get started, we first consider Markov chains.
Since MC are purely stochastic, our problem changes from optimization to computation.
For readability, we thus omit policies from superscripts such as $\totalsum<s, \strategy>$ and write $\totalsum<s>$ instead.

Recall that VaR is the first time step after which $1 - \threshold$ executions have reached the goal.
We can compute this step by iterating the transition relation of the MC, i.e.\ computing where the system is after $n$ steps.
This naturally also gives us the distribution of the remaining executions which have not yet reached the goal.
To obtain the CVaR, we then need to consider the expected cost to reach the goal for this remainder, i.e.\ the classical SSP value.


Formally, fix an MC $\MC$ and goal states $\goalset$.
Let $\nstepprob<n>(s) \coloneqq \Probability[s_n = s \mid \initialstate]$ the probability that the system is in state $s$ after $n$ steps and $\expectedcost(s) \coloneqq \Expectation[\totalsum<s>]$ the expected number of steps to reach a goal state starting in $s$.
We define $\nstepnongoal<n> \coloneqq 1 - \sum_{s \in \goalset} \nstepprob<n>(s)$ the probability of not having reached the goal state after $n$ steps.
Then, $\VaR_\threshold(\totalsum<\initialstate>)$ is the unique value $n$ such that $\nstepnongoal<n+1>{} < \threshold \leq \nstepnongoal<n>$.
By our assumptions, we have that $\nstepnongoal<n> \to 0$ for $n \to \infty$, consequently such an $n$ exists for every $\threshold > 0$.
Finally, let $\nstepexpectedcost<n> \coloneqq {\sum}_{s \in \States} \nstepprob<n>(s) \cdot \expectedcost(s)$ the expected time to reach the goal after $n$ steps.
Note that we can include goal states in the sum as $\expectedcost(s) = 0$ for all goal states and they are absorbing.
Moreover, $\Expectation[\totalsum<s> \mid \totalsum<s> > n] = n + \frac{1}{\nstepnongoal<n>} \nstepexpectedcost<n>$:
We deliberately define $\nstepexpectedcost<n>$ independent of the fraction of runs which already have reached the goal, thus the conditioning of CVaR requires re-weighting.

Together, we obtain an intuitive characterization of $\CVaR$ for SSP, which is the foundation for our solution approaches.
\begin{theorem} \label{stm:cvar_equation}
	For $\VaR_\threshold(\totalsum<\initialstate>) = n$, we have
	\begin{equation*}
		\CVaR_\threshold(\totalsum<\initialstate>) = n + \tfrac{1}{\threshold} \nstepexpectedcost<n>.
	\end{equation*}
\end{theorem}
\begin{proof}
	Inserting the above definitions in \cref{eq:cvar_definition} yields
	\begin{align*}
		\CVaR_\threshold(\totalsum<\initialstate>) & = \tfrac{1}{\threshold} \big( \nstepnongoal<n> \cdot (n + \tfrac{1}{\nstepnongoal<n>} \cdot \nstepexpectedcost<n>) + (\threshold - \nstepnongoal<n>) \cdot n \big) \\
			& = n + \tfrac{1}{\threshold} \nstepexpectedcost<n>. \qedhere
	\end{align*}
\end{proof}
This already yields an effective algorithm for MC:
We compute $\expectedcost$ using standard methods, iteratively compute $\nstepprob<n>$ for increasing $n$ to obtain $\VaR_\threshold(\totalsum<\initialstate>)$, and together get $\CVaR_\threshold(\totalsum<\initialstate>)$.
Unfortunately, $\VaR$ may be of exponential size.
\begin{lemma} \label{stm:var_mc_exponential}
	For every Markov chain $\MC$ we have that $\VaR_\threshold(\totalsum<\initialstate>) \in \mathcal{O}({-}\log \threshold \cdot \cardinality{\States} \cdot p_{\min}^{-\cardinality{\States}})$, where $p_{\min}$ is the minimal transition probability in $\MC$.
	This bound is tight.
\end{lemma}
Thus, our algorithm is exponential.
However, for polynomial VaR the overall algorithm is polynomial, too, since $\expectedcost$ can be computed in polynomial time.
For practical purposes, we can additionally exploit that $\nstepnongoal<n>$ is monotone in $n$ and employ binary search together with exponentiation by squaring.
\begin{lemma} \label{stm:cvar_ssp_polynomial}
	On Markov chains, \textbf{CVaR-SSP} can be solved using polynomially many arithmetic operations.
\end{lemma}
Note that the overall runtime of this algorithm still is exponential, since multiplication itself is not a constant time operation.
In practice, our algorithm can benefit from efficient matrix-multiplication methods and fixed-point arithmetic.

\subsection{Markov Decision Processes}

Now, we move our focus from MC to MDP.
We can re-use some of the observations from the previous section, however the addition of non-determinism complicates the problem significantly.
In particular, there is no unique distribution $\nstepprob<n>$, it rather depends on the chosen policy $\strategy$, which we denote by $\nstepprob<n><\strategy>$.
Analogously, we write $\nstepnongoal<n><\strategy>$ and $\nstepexpectedcost<n><\strategy>$ for the respective values achieved by a given policy $\strategy$ and abbreviate $\CVaR_\threshold(\strategy) \coloneqq \CVaR_\threshold(\totalsum<\initialstate, \strategy>)$ (analogous for $\VaR_\threshold(\strategy)$).
Finally, we write $\expectedcost(s)$ for the optimal expected cost to reach the target starting in $s$, i.e.\ the classical SSP value.

Note that we may have $\VaR_\threshold(\strategy) = \CVaR_\threshold(\strategy) = \infty$ for some $\strategy$.
However, under every \emph{proper} policy $\strategy^p$ the goal is reached within finite time with probability $1$.
Thus $\VaR_\threshold(\strategy^p)$ is bounded, \cref{stm:cvar_equation} remains applicable (by considering the induced MC), and $\CVaR_\threshold(\strategy^p) = n + \frac{1}{\threshold} \nstepexpectedcost<n><\strategy^p> < \infty$ for $n = \VaR_\threshold(\strategy^p)$.
Consequently, the optimal value is finite.

Before diving deeper into the solution concepts, we first prove that an optimal policy always exists.
\begin{theorem} \label{stm:cvar_optimal_policy_exists}
	We have $\CVaR_\threshold^* = \min_{\strategy \in \Strategies} \CVaR_\threshold(\strategy)$.
\end{theorem}
As a naive approach, one thus could try to enumerate all possible policies and apply the reasoning of the previous section.
However, even when only considering memoryless deterministic policies, there may be exponentially many distributions $\nstepprob<n><\strategy>$.
Even worse, the following example shows that optimal policies may require \emph{exponential memory}.
This suggests that enumeration approaches such as policy iteration (another popular approach to solve problems on MDP), or a simple value iteration likely are bound to fail, since both of them typically work with local, \enquote{memoryless} values.
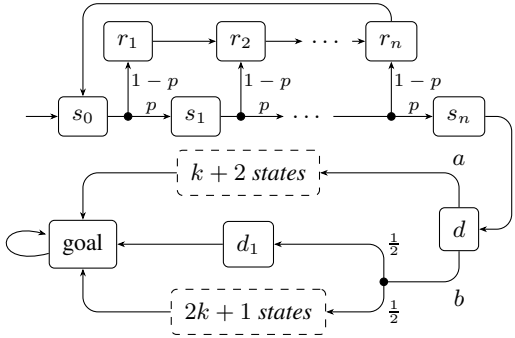
\begin{figure}[t]
	\centering
	\small
	\begin{tikzpicture}[auto,initial text=]
		\node[state,initial left] at (0,0) (init) {$s_0$};
		\node[actionnode] at (0.6,0) (s0a) {};
		\node[state] at (1.5,0) (s1) {$s_1$};
		\node[actionnode] at (2.1,0) (s1a) {};
		\node at (3,0) (dots) {$\dots$};
		\node[actionnode] at (4.1,0) (sdotsa) {};
		\node[state] at (5,0) (sn) {$s_n$};

		\node[state] at (0.6,1) (r1) {$r_1$};
		\node[state] at (2.1,1) (r2) {$r_2$};
		\node        at (3.25,1) (rdots) {$\dots$};
		\node[state] at (4.1,1) (rn) {$r_n$};

		\node[state] at (5,-1.5) (d) {$d$};
		\node[state,dashed] at (2.2,-0.75) (nbstates) {\emph{$k + 2$ states}};
		\node[actionnode] at (4,-2.2) (na) {};
		\node[state,dashed] at (2.2,-2.6) (namany) {\emph{$2k + 1$ states}};
		\node[state] at (2.2,-1.7) (na1) {$d_1$};

		\node[state] at (0,-1.7) (goal) {goal};

		\path[directedge]
			(r1) edge (r2)
			(r2) edge (rdots)
			(rdots) edge (rn)
		;
		\draw[directedge,rounded corners=2mm]
			(rn) -- ([shift={(0,0.5)}]rn.center) -| (init)
		;

		\draw[directedge,rounded corners=2mm]
			(sn) -- ([shift={(0.7,0.0)}]sn.center) |- (d)
		;
		\path[actionedge]
			(init) edge (s0a)
			(s1) edge (s1a)
			(dots) edge (sdotsa)
		;
		\draw[actionedge,rounded corners=2mm]
			(d) |- node[action,anchor=north] {$b$} (na);
		\draw[directedge,rounded corners=2mm]
			(d) |- node[action,anchor=south] {$a$} (nbstates);
		\draw[probedge,rounded corners=2mm]
			(na) |-  node[anchor=west,prob] {$\frac{1}{2}$} (na1);
		\draw[probedge,rounded corners=2mm]
			(na) |- node[anchor=west,prob] {$\frac{1}{2}$} (namany);

		\draw[directedge,rounded corners=2mm]
			(nbstates) -| (goal);
		\draw[directedge,rounded corners=2mm]
			(namany) -| (goal);
		\draw[directedge,rounded corners=2mm]
			(na1) -- (goal);

		\path[probedge]
			(s0a) edge node[prob] {$p$} (s1)
			(s0a) edge[swap,pos=0.6] node[prob] {$1 - p$} (r1)
			(s1a) edge node[prob] {$p$} (dots)
			(s1a) edge[swap,pos=0.6] node[prob] {$1 - p$} (r2)
			(sdotsa) edge node[prob] {$p$} (sn)
			(sdotsa) edge[swap,pos=0.6] node[prob] {$1 - p$} (rn)
		;

		\path[directedge]
			(goal) edge[loop left] (goal)
		;
	\end{tikzpicture}
	\caption{
		Exponential memory may be required.
		The upper part ensures that after $n + 1$ steps, a fraction of $p^n$ executions are in $s_n$ and the remaining $1 - p^n$ are in $s_0$.
		The lower part then comprises a choice between a safe option (action $a$) and a more risky but slightly more efficient option (action $b$).
	} \label{fig:exponential_memory}
\end{figure}
\begin{example}
	Consider the MDP in \cref{fig:exponential_memory} for any $n > 2k + 1$.
	In this case, every $i \cdot (n + 1)$ steps, a \enquote{packet} of $(1 - p^n)^i \cdot p^n$ executions arrives at $d$.
	The optimal choice in $d$ depends on the fraction of executions which are still in the upper part.
	If more than $1 - \threshold$ are still \enquote{on top}, the choice does not matter, since the current execution will surely reach the goal before the VaR, i.e.\ the packet is composed completely of \enquote{good} outcomes and not considered for CVaR.
	If more than $1 - \threshold$ executions already are in the lower part beyond $d$, the optimal choice is $b$, since the current packet only contains \enquote{bad} outcomes; only the expectation counts.
	However, for some $i^*$, the current packet contains exactly those executions which are at the $1 - \threshold$ boundary, i.e.\ containing both good and bad ones.
	Then (for appropriate $n$ and $p$) the optimal choice is action $a$.
	For example, if the current packet is composed of exactly 50\% good and 50\% bad, the expected performance of the bad fraction under $a$ is $k + 2$ compared to $2k + 1$ under $b$.
	One can directly show that the step corresponding to $i^*$ is exponential and thus the policy requires as much memory.
\end{example}
\begin{remark}
	The example above shows that exponential memory is required.
	However, it does not prove that randomization is needed, and we have not found an example where this would be the case.
	We conjecture that deterministic policies actually are sufficient and leave this question for future work.
\end{remark}

Despite the exponential memory requirement, we are able to derive two practical solution techniques, which we explain in the following.
First, we again observe that once the VaR is reached, i.e.\ the system has performed $\VaR_\threshold(\strategy) = n$ steps, we are only interested in the expectation:
Exactly those executions which have not reached the goal after $n$ steps are considered in the expectation computation of CVaR.
\begin{lemma} \label{stm:memoryless_after_var}
	Fix a policy $\strategy$ and let $n = \VaR_\threshold(\strategy)$.
	Then, there exists a policy $\strategy'$ which is stationary after $n$ steps and $\CVaR_\threshold(\strategy') \leq \CVaR_\threshold(\strategy)$.
\end{lemma}
So, intuitively, as before in the Markov chain case, after $n$ steps we are only interested in the optimal expected time $\expectedcost(s)$.
We can use standard methods to compute this value (and a witness policy) in polynomial time \cite[Chp.~7]{DBLP:books/wi/Puterman94}.
This suggests that our task decomposes into two sub-problems, namely
	(i)~reaching the goal quickly, resulting in a small VaR, and
	(ii)~optimally distributing the remaining executions, i.e.\ states with a small expected time to reach the goal $\expectedcost(s)$.
One might feel tempted to first minimize the VaR and then, among \enquote{VaR-optimal} policies, choose one with optimal expected value on the remaining $\threshold$ executions.
However, trying to achieve a VaR as small as possible at all costs may actually come with a significantly larger \enquote{tail} of the distribution -- which is one of the main reasons why VaR is \enquote{seductive, but dangerous}.
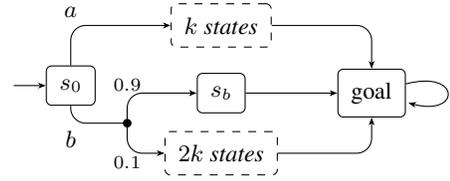
\begin{figure}
	\centering
	\small
	\begin{tikzpicture}[auto,initial text=]
		\node[state,initial] at (0,0) (s0) {$s_0$};
		\node[state,dashed] at (2,0.8) (astates) {\emph{$k$ states}};
		\node[state] at (2,-0.1) (bstate1) {$s_b$};
		\node[state,dashed] at (2,-0.9) (bstates2) {\emph{$2k$ states}};
		\node[actionnode] at (0.75,-0.5) (s0b) {};
		\node[state] at (4,-0.1) (goal) {goal};

		\draw[directedge,rounded corners=2mm]
			(s0) |- node[action,anchor=south] {$a$} (astates);
		\draw[actionedge,rounded corners=2mm]
			(s0) |- node[action,anchor=north] {$b$} (s0b);
		\draw[probedge,rounded corners=2mm]
			(s0b) |- node[prob,anchor=south] {$0.9$} (bstate1);
		\draw[probedge,rounded corners=2mm]
			(s0b) |- node[prob,anchor=north] {$0.1$} (bstates2);
		
		\draw[directedge,rounded corners=2mm]
			(astates) -| (goal);
		\draw[directedge]
			(bstate1) -- (goal);
		\draw[directedge,rounded corners=2mm]
			(bstates2) -| (goal);

		\path[directedge]
			(goal) edge[loop right] (goal)
		;
	\end{tikzpicture}
	\caption{
		VaR optimization is suboptimal.
	} \label{fig:var_optimization}
\end{figure}
\begin{example} \label{ex:minimizing_var_is_bad}
	Consider the MDP in \cref{fig:var_optimization} with $\threshold = 0.15$.
	Action $b$ is strictly preferred both for expectation as well as VaR optimization, while action $a$ is CVaR-optimal.
\end{example}
More strikingly, observe that action $b$ in the example is VaR optimal even if instead of $2k$ states there would be \emph{arbitrarily many}.
So, in a sense, VaR does not consider the $\threshold$ worst outcomes but rather the $1 - \threshold$ best.
By optimizing the first $1 - \threshold$ executions, the remaining portion may be positioned disproportionately bad.
Instead, we have to balance between issues (i) and (ii), and search for a \emph{trade-off} between a small VaR and the distribution of the remaining executions.
We present two different approaches to find this trade-off, based on linear programming and value iteration, respectively.
\subsection{Linear Programming}
\begin{figure}[t]
	\begin{gather*}
		\min {\sum}_{s \in \States} p_{s,n} \cdot \expectedcost(s) \text{ subject to} \\[0.5em]
		\text{All variables non-negative} \\[0.5em]
		p_{\initialstate,0} = 1 \qquad p_{s,0} = 0 \quad \forall s \in \States, s \neq \initialstate \\
		\begin{aligned}
			p_{s,i} & = {\sum}_{a \in \stateactions(s)} p_{s,a,i} \quad \forall s \in \States, i < n \\
			p_{s',i + 1} & = {\sum}_{s \in \States, a \in \stateactions(s)} p_{s,a,i} \cdot \mdptransitions(s, a, s') \quad \forall s' \in \States, i < n \\
		\end{aligned} \\[0.5em]
		{\sum}_{s \in \goalset} p_{s,n - 1} \leq 1 - \threshold \leq {\sum}_{s \in \goalset} p_{s,n}
	\end{gather*} %
	\caption{LP to compute $\CVaR_\threshold$ given VaR guess $n$.} \label{fig:cvar_lp}
\end{figure}
For our linear programming approach, first assume that we are magically given the optimal VaR $n$, i.e.\ the VaR which allows us to obtain the optimal CVaR.
Inspired by \cite[Fig.~3]{DBLP:journals/lmcs/ChatterjeeKK17}, we can construct an LP of size linear in $n$, where the set of solutions corresponds to all policies achieving this VaR, shown in \cref{fig:cvar_lp}.
This LP then computes the minimal CVaR over these policies.
Intuitively, the LP is obtained by \emph{unrolling} the MDP until step $n$ and building reachability constraints on that MDP.
\begin{theorem} \label{stm:cvar_lp_solution}
	If there exists a policy $\strategy$ such that $\VaR_\threshold(\strategy) = n$ and $\CVaR_\threshold(\strategy) = C$, the LP in \cref{fig:cvar_lp} has a solution with value at most $\threshold \cdot (C - n)$.
	If the LP in \cref{fig:cvar_lp} has a solution for some $n$ with value $E$, there exists a policy achieving $\CVaR_\threshold(\strategy) = n + \frac{1}{\threshold} E$.
\end{theorem}
\begin{proof}[Prook sketch (see Appendix)]
	\underline{First part}:
	We construct an assignment to the LP's variables:
	Set $p_{s, i} = \nstepprob<i><\strategy>(s)$ and $p_{s, a, i} = \Probability[s_i = s, a_i = a \mid \initialstate, \strategy]$.
	This assignment satisfies the first three constraints of the LP.
	For the fourth constraint, observe that by $\VaR_\threshold(\strategy) = n$, we have that $\sum_{s \in \goalset} \nstepprob<n - 1><\strategy>(s) < 1 - \threshold \leq \sum_{s \in \goalset} \nstepprob<n><\strategy>(s)$.
	By \cref{stm:cvar_equation}, we have that $C = n + \tfrac{1}{\threshold} \nstepexpectedcost<n><\strategy>$.
	Since $\nstepexpectedcost<n><\strategy> = \sum_{s \in \States} \nstepprob<n><\strategy>(s) \cdot \expectedcost(s) = \sum_{s \in \States} p_{s, n} \cdot \expectedcost(s)$, we get that $\nstepexpectedcost<n><\strategy> = \threshold \cdot (c - n)$, proving the claim.

	\underline{Second part}: We construct the policy $\strategy$ as follows.
	For the first $n$ steps, at step $i$ in state $s$, choose action $a$ with probability $p_{s, a, i}$.
	Afterwards, i.e.\ starting from step $n$, in state $s$ follow a policy achieving the optimal expected cost $\expectedcost(s)$ (note the similarity to \cref{stm:memoryless_after_var}).
	Clearly, $\nstepprob<i><\strategy>(s) = p_{s, i}$ and thus $v = \nstepexpectedcost<n><\strategy>$.
	We have $\VaR_\threshold(\strategy) = n - 1$ if $\sum_{s \in \goalset} p_{s, n - 1} = 1 - \threshold$ and $\VaR_\threshold(\strategy) = n$ otherwise.
	In both cases, we can prove $\nstepexpectedcost<n - 1><\strategy> = \threshold + \nstepexpectedcost<n><\strategy>$.
	Inserting yields $\CVaR_\threshold(\strategy) = (n - 1) + \frac{1}{\threshold} \nstepexpectedcost<n - 1><\strategy> = n + \frac{1}{\threshold} \nstepexpectedcost<n><\strategy> = n + \frac{1}{\threshold} v$.
\end{proof}
This directly suggests an algorithm:
We simply try each possible VaR value and solve the associated LP.
Unfortunately, as suggested by \cref{stm:var_mc_exponential}, the VaR obtained by CVaR optimal policies may be at least exponential.
We prove a matching upper bound in the following.
\begin{lemma} \label{stm:var_bound_from_cvar}
	Fix a policy $\strategy$ and let $\strategy^*$ be an optimal policy.
	Then $\VaR_\threshold(\strategy^*) \leq \CVaR_\threshold(\strategy)$.
\end{lemma}
Note that the VaR of CVaR optimal policies may not be optimal, i.e.\ potentially $\VaR_\threshold(\strategy^*) > \VaR_\threshold(\strategy)$, see \cref{ex:minimizing_var_is_bad}.

Furthermore, by employing our assumption that a proper policy exists and combining it with \cref{stm:var_mc_exponential}, we can find a policy with CVaR of bounded size.
\begin{lemma} \label{stm:mdp_cvar_exponential_bound}
	There is a policy with at most exponential CVaR.
\end{lemma}
\Cref{stm:var_bound_from_cvar,stm:mdp_cvar_exponential_bound} together then yield the desired result. 
\begin{corollary} \label{stm:mdp_var_exponential_bound}
	$\VaR_\threshold(\strategy^*)$ is at most exponential.
\end{corollary}
This implies that our LP algorithm is EXPTIME: We solve exponentially many linear programs of exponential size.
However, \cref{stm:var_bound_from_cvar} also yields a dynamic \enquote{stopping criterion} for our algorithm:
We do not always need to try out all exponentially many possible values for VaR.
Instead, once we found a solution, we can use the CVaR obtained by this solution as new upper bound for the VaR guesses.
Thus, the exponential time solely depends on the magnitude of $\CVaR_\threshold(\strategy^*)$.
In particular, if $\CVaR_\threshold(\strategy^*)$ is of polynomial size, our algorithm is PTIME, since we can stop after polynomially many steps.

We conclude this section with a series of remarks, putting our results into context.
\begin{remark}
	We conjecture that these results are optimal, i.e.\ that \textbf{CVaR-SSP} is EXPTIME-complete.
	However, a proof of hardness seems surprisingly difficult:
	Only recently, \cite{DBLP:conf/icalp/BalajiK0PS19} proved that the (conceptually much simpler) problem of finite-horizon reachability is EXPTIME-complete; a question open for several decades, posed already by \cite{DBLP:journals/mor/PapadimitriouT87}.
	Observe that deciding whether a policy $\strategy$ exists such that $\VaR_\threshold(\strategy) \leq n$ is a special case of finite-horizon reachability.
	Unfortunately, the techniques of \cite{DBLP:conf/icalp/BalajiK0PS19} are not applicable to our case, since we assumed the existence of proper policies.
	We conjecture that the associated VaR problem nevertheless is EXPTIME-complete, too.
	Yet, even proving hardness of the VaR problem does \emph{not} immediately prove that \textbf{CVaR-SSP} itself is EXPTIME-complete:
	There might be an algorithm which can determine the CVaR without explicitly determining the VaR.
	It however seems unlikely that CVaR can be accurately computed without knowledge of VaR.
\end{remark}
\begin{remark}
	In \cite{DBLP:conf/lics/KretinskyM18}, the authors also employed a \enquote{VaR-guessing} approach for a similar problem.
	In their setting however, only linearly many possible value for VaR exist, thus yielding a polynomial algorithm.
	See \cite[Thm.~3.33]{piribauer2021} for a translation of our case to the scenario of \cite{DBLP:conf/lics/KretinskyM18} and an alternative proof for the exponential upper bound.
\end{remark}

\begin{remark}
	Our LP approach can be adapted to the constrained variant, i.e.\ answer the question \enquote{given that CVaR should be at least $x$, what is the maximal expectation?}, by changing the CVaR objective to an appropriate constraint and adding expectation maximization as objective.
\end{remark}
\subsection{Value Iteration}

While LP is appealing in theory due to its polynomial complexity, in practice it often is outperformed by approaches such as VI, despite worse theoretical complexity.
We now discuss several insights which ultimately lead to a VI algorithm for our problem, yielding precise results.
This is particularly intriguing since a VI approach for the similar scenario of \cite{DBLP:conf/lics/KretinskyM18} remains elusive.

To derive a VI approach, we require an iteration operation which yields a value for each state based on the values of their respective successors.
So, suppose we naively want to compute $\CVaR_\threshold$ of a state $s$ based on the $\CVaR_\threshold$ of its successors.
Unfortunately, we cannot simply combine the $\CVaR_\threshold$ of the successors:
For example, it might be the case that all \enquote{bad} outcomes (i.e.\ those worse than the $\VaR_\threshold$) are all those which move to one particular successor.
Hence, we would need $\CVaR_1$ for that successor and $\CVaR_0$ for all others.
Consequently, we need to employ a different approach.

By definition of $\CVaR_\threshold$, precisely a fraction of $\threshold$ executions starting in $s$ are bad ones, and these have to distribute \emph{somehow} over the successors.
Thus, there is a weighting of successors, reflecting how the bad executions distribute.
\begin{lemma} \label{stm:cvar_weight_decomposistion}
	Let $s$ be a state, $a \in \stateactions(s)$ an available action, and $\strategy$ a policy.
	There exist weights $w : \States \to [0, 1]$ such that $\sum_{s' \in \States} \mdptransitions(s, a, s') \cdot w(s') = \threshold$ and
	\begin{multline*}
		\CVaR_\threshold(\totalsum<s, \strategy>) = \\
		1 + \sum_{{a \in \stateactions(s), s' \in \States}} \strategy((s), a) \mdptransitions(s, a, s') \CVaR_{w(s')}(\totalsum<s', \strategy^a>),
	\end{multline*}
	where $(s)$ denotes a path comprising only $s$ and $\strategy^a$ denotes the policy $\strategy$ after taking action $a$.
\end{lemma}
\begin{proof}
	Follows directly from the above discussion.
\end{proof}
Note the similarity of the above equation to a Bellman update:
For known/fixed weights $w$, we would obtain a regular value iteration.
Unfortunately, these weights globally depend on the policy; even the decision in a state $s'$ which is not reachable from state $s$ influences the weight distribution in $s$.
Nevertheless, we can use the underlying insights to derive a value iteration approach.

\subsection{Characterization through Pareto Sets}

In light of \cref{stm:cvar_weight_decomposistion}, we do not want to compute the CVaR for a single threshold $\threshold$, but for \emph{all} thresholds $0 < \threshold \leq 1$.
So, essentially, we aim to answer the question \enquote{given a threshold of $\threshold$, what is the best achievable CVaR?} for all thresholds and all states.
We approach this question with a new, novel perspective.
In particular, we propose to answer a different question, namely, \enquote{for an arbitrary step-bound $n$, given that at least $1 - \threshold$ executions have to reach the goal within $n$ steps, what is the best expected time to reach the goal after $n$ steps?}, a trade-off which can be described by a \emph{Pareto set}.
Quite surprisingly, we can (i)~derive CVaR from such a set and (ii)~compute these sets for increasing $n$ using VI.

\begin{definition}
	Let $s \in \States$ a state and $n \in \Naturals$ a step bound.
	We define the \emph{SSP Pareto set} $\pareto<n><s> \subseteq [0, 1] \times \Reals_{\geq 0}$ where $(p, E) \in \pareto<n><s>$ iff there exists a policy $\strategy$ such that $\nstepnongoal<n><\strategy> \leq 1 - p$ and $\nstepexpectedcost<n><\strategy> \leq E$, i.e.\ 
		(i)~the probability to reach a goal state within $n$ steps starting in $s$ is at least $p$, and
		(ii)~after $n$ steps, the expected time to reach goal states is at most $E$.
\end{definition}
%
%
\begin{lemma} \label{stm:cvar_from_pareto}
	For every $(1 - \threshold, E) \in \pareto<n><\initialstate>$ with witness policy $\strategy$, we have $\CVaR_\threshold(\strategy) \leq n + \frac{1}{\threshold} E$.
	For every policy $\strategy$, we have $(1 - \threshold, \threshold \cdot (\CVaR_\threshold(\strategy) - n)) \in \pareto<n><\initialstate>$ where $n = \VaR_\threshold(\strategy)$.
\end{lemma}
To obtain an algorithm based on this idea, we need an effective procedure to compute $\pareto<n><s>$.
To this end, we show that $\pareto<n><s>$ is a convex polygon where vertices correspond to deterministic policies, and show how $\pareto<n><s>$ can be computed using a Bellman-style iteration.
\begin{lemma} \label{stm:pareto_shape}
	The set $\pareto<n><s>$ is an upward and leftward closed, convex polygon where all vertices correspond to Markovian deterministic policies.
\end{lemma}
\begin{proof}[Proof sketch (see Appendix)]
	Closure and Convexity follow directly.
	We prove the polygon-claim by induction on $n$.

	For $n = 0$, we either have that $\pareto<0><s> = [0, 1] \times \Reals_{\geq 0}$ if $s \in \goalset$, or, if $s \notin \goalset$, $\pareto<0><s> = \{0\} \times [e(s), \infty)$, both by definition.
	In both cases, $\pareto<0><s>$ is a polygon with the extremal points $(1, 0)$ and $(0, e(s))$, respectively, obtained by stationary policies.

	For the induction step, fix $n$ and a state $s$.
	We prove that $(p, E) \in \pareto<n+1><s>$ iff there exist a distribution $w : \stateactions(s) \to [0, 1]$ and achievable points $(p_{a, s'}, E_{a, s'}) \in \pareto<n><s'>$ 
	such that
	\begin{equation*}
		(p, E) = {\sum}_{a \in \stateactions(s), s' \in \States} w(a) \cdot \mdptransitions(s, a, s') \cdot (p_{a, s'}, E_{a, s'})
	\end{equation*}
	This equality follows from the interpretation of the Pareto set:
	For all actions $a$ and successors $s'$ there exists a policy $\strategy_{a, s'}$ such that after $n$ steps at least a fraction of $p_{a, s'}$ executions have reached the goal and the expected time to reach the goal is at most $E_{a, s'}$.
	So, we can take one step and then simply follow these respective policies to achieve the values in the equation.
	Dually, if there is a policy $\strategy$ for $(p, E) \in \pareto<n+1><s>$, we immediately get policies achieving the respective values in the successors (note the similarity to regular value iteration).
	The claim follows by the hypothesis.
\end{proof}
This proof also yields an effective way of computing $\pareto<n><s>$.
\begin{corollary} \label{stm:pareto_computation}
	We have that
	\begin{equation*}
		\pareto<n + 1><s> = \conv\left({\Union}_{a \in \stateactions(s)} {\bigoplus}_{s' \in \States} \mdptransitions(s, a, s') \cdot \pareto<n><s'>\right),
	\end{equation*}
	where $\oplus$ is the Minkowski sum and $\conv$ the convex hull.
\end{corollary}
%
%
With these results, we are ready to present our value iteration approach in \cref{alg:cvar_reachability}.
As expected, it computes $\pareto<n><s>$ for increasing $n$ using \cref{stm:pareto_computation} and derives the optimal obtainable CVaR assuming that the VaR is at most $n$ using \cref{stm:cvar_from_pareto}.
On top, the algorithm uses \cref{stm:var_bound_from_cvar} as stopping criterion, ultimately yielding the optimal CVaR.
Note that the algorithm can compute the optimal CVaR for several thresholds $\threshold$ simultaneously at essentially no additional cost:
While computing the CVaR for the smallest threshold, the CVaR corresponding to the other thresholds can be obtained as an intermediate result.
\begin{algorithm}[t]
	\caption{Value Iteration to compute CVaR} \label{alg:cvar_reachability}
	\begin{algorithmic}
		\Require MDP $\MDP$, threshold $\threshold$
		\Ensure Optimal $\CVaR_\threshold$

		\State $\mathtt{c} \gets \infty$, $n \gets 0$
		\While{$n \leq \mathtt{c}$}
			\State Compute $\pareto<n><s>$ for all $s \in \States$
			\State $\mathtt{c}_n \gets n + \frac{1}{\threshold} \cdot \left( \min \{E \mid (1 - \threshold, E) \in \pareto<n><\initialstate>\} \union \{ \infty \} \right)$
			\State $\mathtt{c} \gets \min(\mathtt{c}, \mathtt{c}_n)$, $n \gets n + 1$
		\EndWhile
		\Return $\mathtt{c}$
	\end{algorithmic}
\end{algorithm}
\begin{theorem} \label{stm:cvar_algo_correct}
	\Cref{alg:cvar_reachability} is correct, i.e.\ always terminates and returns the optimal CVaR.
	Moreover, it is EXPTIME.
\end{theorem}
%

\subsection{Total Cost}

To extend our approach to the general scenario of total cost, only minor adjustments are necessary.
We omit the completely analogous proofs of correctness.
Note that both cases show that instead of Markovian policies, we now require policies which (only) depend on the total incurred cost.

\paragraph{Linear Programming}
Recall that in order to obtain the LP approach we effectively \enquote{unrolled} the MDP, augmenting the state space with a step counter.
We change this counter to track the accumulated cost, i.e.\ require
\begin{equation*}
	p_{s', i} = {\sum}_{s \in \States, a \in \stateactions(s)} p_{s, a, i - \actioncost(s, a)} \cdot \mdptransitions(s, a, s')
\end{equation*}
for $i < n$, adapting appropriately at the boundary.
Informally, $p_{s, a, i}$ now corresponds to the probability of taking action $a$ in state $s$ when the total accumulated cost so far is $i$.

Around the VaR-guess $n$, additional care is needed:
For example, suppose we are at state $s$ with an accumulated cost of $n - 1$ and take action $a$ with cost $\actioncost(s, a)$.
In the LP, we would thus consider the variable $p_{s', n - 1 + \actioncost(s, a)}$.
Now, for $\actioncost(s, a) > 1$, we need to modify the objective to also consider this part of the flow.
In particular, we change the objective function to $\sum_{s \in \States, 0 \leq c < C_{\max}} p_{s, n + c} \cdot (e(s) + c)$, where $C_{\max} = \max \actioncost(s, a)$ is the maximal cost.

\paragraph{Value Iteration}
The Pareto set describes a trade-off between probability of reaching and expected remaining steps at $n$ steps.
We again change the interpretation of $n$ to a cost bound.
This means that a point $(p, E)$ is in $\pareto<n><s>$ iff the goal can be reached with probability at least $p$ while incurring a cost of at most $n$ and at the same time have an expected cost of at most $E$ to reach the goal afterwards.
All statements can be replicated analogously, in particular we arrive at
\begin{equation*}
	\pareto<n><s> = \conv\left({\Union}_{a \in \stateactions(s)} {\bigoplus}_{s' \in \States} \mdptransitions(s, a, s') \cdot \pareto<n - \actioncost(s, a)><s'>\right),
\end{equation*}
where $\pareto<n><s> = \emptyset$ for $n < 0$.
Consequently, we need to store the last $C_{\max}$ Pareto sets for each state in the VI algorithm.

\paragraph{Zero Cost States}
We assumed that there are no zero-cost actions for simplicity.
This assumption is implicitly applied in the above arguments.
We can however \enquote{inline} such states through additional pre-computation:
Intuitively, for a zero-cost action $a$, we can compute all possible ways the system could evolve after choosing $a$ while restricted to using only zero-cost actions, and then replace the outcome of $a$ with these options.
The technical details of this procedure however are quite involved and beyond the scope of this work.
\begin{table*}[t]
	\caption{Summary of our experiments.
		For each model we list, from left to right, the number of states, actions, and transitions, the expected total cost until goal states are reached, i.e.\ the classical SSP value, the considered threshold, resulting VaR and CVar, and finally the times required to compute the SSP values, CVaR via LP, and CVaR via VI, respectively.
		\emph{MO} denotes a memout.
		To ease presentation, we only considered a threshold of $\threshold = 10\%$ for this experiment.
	} \label{tbl:experiments}
	\centering %
	\newcommand{\cc}[1]{\multicolumn{1}{c}{#1}}
	\begin{tabular}{lrrrrrrrrr}
		Model                   & \cc{$\cardinality{\States}$} & \cc{$\cardinality{\Actions}$} & \cc{$\cardinality{\mdptransitions}$} & \cc{$\Expectation$} & \cc{$\VaR_{10\%}$} & \cc{$\CVaR_{10\%}$} & \cc{SSP} &   \cc{LP} & \cc{VI} \\
		\midrule
		\textbf{Grid}($x = 4$)  &                      1{,}270 &                       4{,}230 &                             12{,}390 &                 6.7 &                  8 &                11.0 &       0s &        2s &      0s \\
		\textbf{Grid}($x = 8$)  &                      3{,}196 &                      12{,}796 &                             37{,}996 &                13.1 &                 16 &                17.6 &       1s &      199s &      2s \\
		\textbf{Grid}($x = 16$) &                      6{,}396 &                      28{,}796 &                             85{,}996 &                19.3 &                 22 &                23.6 &       2s &  2{,}373s &     11s \\
		\textbf{Grid}($x = 32$) &                     12{,}796 &                      60{,}796 &                            181{,}996 &                35.9 &                 39 &                40.8 &      15s & $>$1h &    221s \\
		\midrule
		\textbf{FireWire}       &                    138{,}130 &                     302{,}654 &                            304{,}826 &               166.2 &                167 &               167.0 &       3s & \emph{MO} &      3s \\
		\textbf{WLAN}           &                     87{,}345 &                     157{,}457 &                            177{,}639 &                48.0 &                 61 &                62.3 &       1s & \emph{MO} &      1s
	\end{tabular}
\end{table*}

\section{Evaluation}

We implemented prototypes of our algorithms in Java (Oracle JVM 17.0.1), delegating LP calls to \texttt{Gurobi} 9.1.2, running on consumer-grade hardware (AMD Ryzen 5 3600, 3.60 Ghz, 16 GB RAM).
The JVM is limited to 10 GB of RAM through \texttt{-Xmx10G}.
We augmented the VI approach by parallel computation and implemented a tailored Minkowski sum / convex hull computation.
The LP approach computes a lower bound on the minimal VaR via a reachability VI.
These optimizations already yield order-of-magnitude improvements compared to a naive implementation.
Our implementation, all models, and instructions to reproduce the experiments can be found at \url{https://doi.org/10.5281/zenodo.5764140}\nocite{meggendorfer_tobias_2021_5764141}.

Since \cite{DBLP:conf/nips/ChowTMP15} solve a similar problem (albeit with discounting), we tried evaluating their approach with a sufficiently large discounting factor, too.
Unfortunately, we could not obtain an implementation of their methods.

\begin{figure}[t]
	\centering
	\begin{tikzpicture}[auto,scale=0.8]
		\draw[step=1.0,darkgray] (0,0) grid (4,4);
		\node[anchor=center] at (3.5,2.5) {\faFan};
		\node[anchor=center] at (1.5,1.5) {\faTrash};

		\node[anchor=center] at (0.5,0.5) (droid) {\faRobot};
		\node[anchor=center] at (2.5,1.5) (janitor) {\faBroom};
		\node[anchor=center] at (3.5,3.5) {\faHome};

%
%

	\end{tikzpicture}
	\caption{
		Visual representation of \textbf{Grid} for size 4x4.
		The robot currently is at $(1,1)$, while the janitor is at $(3,2)$.
	} \label{fig:gridworld}
\end{figure}
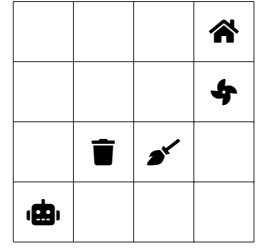

\paragraph{Models}
To test out our methods, we first consider a hand-crafted probabilistic path planning task on a grid world, called \textbf{Grid}.
A robot has to move to its charging station, avoiding fixed obstacles.
Moreover, a janitor is moving semi-randomly through a part of the region.
In particular, the janitor faces a direction and either moves into that direction or randomly turns to one side.
Whenever the robot is close to the janitor (distance $\leq$ 1), it is not allowed to move in order to avoid being stepped on.
We consider the problem for several grid widths to investigate scalability.
In order to keep the probability of interaction with the janitor high, we always restrict the janitor to a 4x4 grid that the robot necessarily has to cross.
Furthermore, we consider two models from the literature, namely \textbf{FireWire} \cite{DBLP:journals/fac/KwiatkowskaNS03}, the IEEE 1394 \enquote{FireWire} root contention protocol, and \textbf{WLAN} \cite{DBLP:conf/papm/KwiatkowskaNS02}, the CSMA/CA mechanism of the 802.11 Wireless LAN protocol.
See \cite{DBLP:journals/fmsd/KwiatkowskaNPS06} for further details on how \textbf{FireWire} and \textbf{WLAN} are constructed.

To evaluate the influence of the threshold, we consider another hand-crafted model \textbf{Walk}:
Here, the system moves along a line of length $n$ and can at each position choose to take one step forward, succeeding with 50\% probability, or gamble to double its current position.
The doubling action has a 10\% chance to fail, instead halving the current position, and is disabled if it fails thrice.
Note that the risk of the gamble, i.e.\ how much is \enquote{lost} in case of a fail, changes with the current position of the system.
As such, depending on the level of risk aversion, the system may choose the gambling option at different positions (or even not at all).

For simplicity, we consider uniform cost for all models.

\paragraph{Results}
Our results for the first experiment are summarized in \cref{tbl:experiments}.
We clearly see that the LP approach quickly becomes infeasible, while VI can tackle significantly larger models.
This is in line with the usual observations, where LP is more appealing in theory, but VI scales much better in practice.
We highlight that the time required by VI is comparable to the time needed to simply optimize SSP.
While solving the SSP first is required by our methods, it nevertheless is encouraging that the overhead of risk-aware optimization is not too large in these cases. 

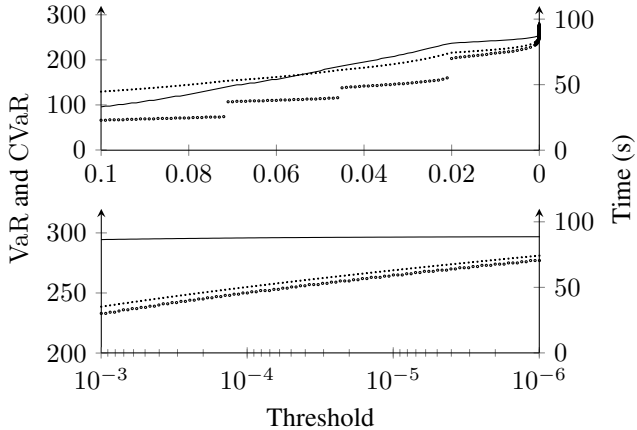
\begin{figure}[t]
	\centering
	\begin{tikzpicture}
		\begin{axis}[xmin=0, xmax=0.1, ymin=0, ymax=320, x dir=reverse,
				height=3.5cm, width=0.85 \columnwidth,
				hide x axis, axis y line=left, ylabel={$\VaR$ and $\CVaR$}, ylabel style={at=(ticklabel cs:-0.2)}]
			\addplot[only marks,mark=o,mark options={scale=0.2}] table [x index=0, y index=1, col sep=comma] {figure_data_random.csv};
			\addplot[only marks,mark=x,mark options={scale=0.2}] table [x index=0, y index=2, col sep=comma] {figure_data_random.csv};
		\end{axis}
		\begin{axis}[xmin=0, xmax=0.1, ymin=0, ymax=110, x dir=reverse,
				height=3.5cm, width=0.85 \columnwidth,
				xticklabel style={
					/pgf/number format/fixed,
					/pgf/number format/precision=5
				},
				axis x line=bottom, axis y line=right, xlabel=\empty, ylabel={Time (s)}, ylabel style={at=(ticklabel cs:-0.2)}, x axis line style={-}, legend pos=north west]
			\addplot[no marks] table [x index=0, y index=3, col sep=comma] {figure_data_random.csv};
		\end{axis}
		
		\begin{scope}[yshift=-2.7cm]
			\begin{axis}[xmode=log,xmin=0.0000001, xmax=0.001, ymin=200, ymax=320, x dir=reverse,
					height=3.5cm, width=0.85 \columnwidth,
					hide x axis, axis y line=left]
				\addplot[only marks,mark=o,mark options={scale=0.2}] table [x index=0, y index=1, col sep=comma] {figure_data_random.csv};
				\addplot[only marks,mark=x,mark options={scale=0.2}] table [x index=0, y index=2, col sep=comma] {figure_data_random.csv};
			\end{axis}
			\begin{axis}[xmode=log,xmin=0.000001, xmax=0.001, ymin=0, ymax=110, x dir=reverse,
					height=3.5cm, width=0.85 \columnwidth,
					xticklabel style={
						/pgf/number format/fixed,
						/pgf/number format/precision=5
					},
					axis x line=bottom, axis y line=right, xlabel={Threshold}, ylabel=\empty, x axis line style={-}, legend pos=north west]
				\addplot[no marks] table [x index=0, y index=3, col sep=comma] {figure_data_random.csv};
			\end{axis}
		\end{scope}
	\end{tikzpicture}
	\caption{
		Runtime evaluation of our VI approach on the \textbf{Walk} model for different thresholds.
		We also depict the VaR and CVaR for each threshold by dots.
		Note that thresholds are \emph{decreasing} from left to right.
		For readability, we include a zoomed, logarithmic plot for thresholds from $10^{-3}$ to $10^{-6}$.
	} \label{fig:cvar_evolution}
\end{figure}

The results of the second experiment, evaluating the influence of the threshold, are depicted in \cref{fig:cvar_evolution}.
We omitted evaluating the LP approach here, since it took over 30 minutes to evaluate a single threshold.
We clearly see the points where the optimal strategy switches away from taking the risky doubling action by a sharp increase of $\VaR$.
Moreover, the VaR and CVaR increase mostly linearly with the threshold, only spiking exponentially for very small thresholds.
This is to be expected due to \cref{stm:var_mc_exponential}: a small portion of probability mass remains inside the system for a long time.
However, the time required for the VI steps decreased drastically, since the number of points in $\pareto<n><s>$ decreased.
After this many steps, only a single dominant strategy remains, and most Pareto sets actually are singletons.
Altogether, we observe that the runtime of VI seems to depend mostly linearly on the threshold, even on an adversarially crafted model.

\paragraph{Improvements}
As our simple optimization heuristics already yielded significant improvements, there likely are many further possibilities.
We conjecture that additional structural properties might be used to speed up computation of $\pareto<n><s>$, e.g.\ a special structure of optimal policies.
Moreover, we found that the performance of VI improves if we merge extremal points of $\pareto<n><s>$ which are, for example, very close to each other or lie just on boundary of the convex hull (i.e.\ removing them barely changes $\pareto<n><s>$).
Since Bellman operators typically are contractive, we conjecture that the error introduced by this merging can be bounded, allowing for a trade-off between precision and speed.
More generally, we think that in order to achieve a given precision of $\varepsilon$, polynomially in $\frac{1}{\varepsilon}$ and $\log \threshold$ many points for $\pareto<n><s>$ may be sufficient.
Here, the ideas of \cite{DBLP:conf/focs/PapadimitriouY00} could be applicable.

\section{Conclusion}
We have presented a new risk-aware perspective on stochastic shortest path through the lens of CVaR.
For this objective, we have derived an LP and a VI based solution, both of which yield precise, provably correct results.
This analysis naturally comes at an additional price, however our experiments show that already with a simple implementation, our approach is feasible on moderately complex problems.

For future work, we aim to provide tight complexity bounds for \textbf{CVaR-SSP}.
In \cite{DBLP:journals/mor/Bonet07}, a rather general condition for polynomial convergence of VI for SSP is presented, which might be applicable to our approach, too.
Moreover, we seek to study the exact structure of optimal policies.
In particular, we conjecture that they do not alternate between actions.
For the practical side, we plan to investigate the improvements mentioned in the previous section as well as study the influence of fixed-precision rounding errors.
Finally, we want to investigate how risk-aware policies differ from purely expectation maximizing solutions in practice.
\paragraph{Acknowledgements} We thank Anna Lukina for the discussion sparking the initial idea.
Moreover, we thank the anonymous reviewers for their insightful comments.
In particular, after investigating some of their remarks, we found a substantial improvement of the VI algorithm, increasing its performance by an exponential factor.

\bibliography{main}

\clearpage
\appendix
\section{Technical Appendix}

\subsection{Proofs -- Markov Chains}
\begin{proof}[Proof of \cref{stm:var_mc_exponential}]
	\underline{Upper bound}: By assumption, the probability to reach the goal is $1$ for every state of the chain.
	Through \cite[Lemma~5.1]{DBLP:journals/jacm/BrazdilKK14} we get that $\nstepnongoal<n> \leq 2 c^n$ for large $n$, where $c = \exp(-\cardinality{\States}^{-1} p_{\min}^{\cardinality{\States}})$.
	By solving $2 c^n = \threshold$ for $n$ we obtain the result.

	\underline{Lower bound}: Consider the upper part of the MDP in \cref{fig:exponential_memory}, i.e.\ all states $s_i$ and $r_i$, and state $d$, a Markov chain.
	Moreover, let $\goalset = \{d\}$.
	After $n+1$ steps, $p^n$ executions are in $d$ and the rest is back in state $s_0$.
	To have at least $\threshold$ many executions in $d$, we thus require at least $m = \log\threshold / \log(1 - p^n)$ such rounds with length $n + 1$.
	Consequently, we need at least $- \log\threshold \cdot (n + 1) \cdot p^{-n}$ steps, as $\log(1 - x) \leq -x$ for $x \in (0, 1)$.
\end{proof}
\begin{proof}[Proof sketch of \cref{stm:cvar_ssp_polynomial}]
	Let $P$ denote the transition matrix of the Markov chain.
	Then, $\nstepprob<n> = P^n \cdot e_1$ where $e_1$ is the unit vector corresponding to the initial state.
	By setting $P_0 = P$ and iterating $P_{k + 1} = P_k \cdot P_k$, we get $P_k = P^{2^k}$.
	We repeat this process until we obtain that $\nstepnongoal<{2^{k+1}}>{} < \threshold$.
	By \cref{stm:var_mc_exponential}, this requires only polynomially many steps, each of which comprises a matrix multiplication, which again amounts to polynomially many operations.
	Then, we know that $\VaR_\threshold$ lies between $2^k$ and $2^{k+1}$.
	(Note that the size of the entries in $P_k$ may grow exponentially, hence the overall time complexity is exponential if the algorithm is implemented with arbitrary precision arithmetic.)
	We repeat this process recursively by computing $P_k \cdot P_0 \cdot e_1$, $P_k \cdot P_1 \cdot e_1$ etc., similar to a binary search for VaR.
	After polynomially many steps, we obtain the precise VaR together with the distribution of the remaining executions.
	Finally, we compute $\expectedcost(s)$ for all states in polynomial time and together obtain CVaR by \cref{stm:cvar_equation}.
\end{proof}

\subsection{Proofs -- Markov Decision Processes}

\begin{proof}[Proof of \cref{stm:cvar_optimal_policy_exists}]
	For every $n$, let $\Strategies_n$ the set of policies achieving $\VaR_\threshold(\strategy) \leq n$.
	Clearly, $\Strategies_n \subseteq \Strategies_{n + 1}$.
	Moreover, $\Strategies_n$ is exactly the set of policies reaching the goal states with probability at least $1 - \threshold$ within $n$ steps---a closed set by optimality of deterministic policies \cite[Chap.~4]{DBLP:books/wi/Puterman94}.
	Next, we define $\Strategies_n' = \Strategies_n \setminus \Strategies_{n - 1}$.
	We have that $\{\strategy \in \Strategies \mid \CVaR_\threshold(\strategy) < \infty\} \subseteq \Union_{n \in \Naturals_0} \Strategies_n = \Union_{n \in \Naturals_0} \Strategies'_n$.
	Consequently, there exists an $n$ such that
	\begin{equation*}
		\CVaR_\threshold^* = {\inf}_{\strategy \in \Strategies'_n} \CVaR_\threshold(\strategy) = n + \tfrac{1}{\threshold} {\inf}_{\strategy \in \Strategies'_n} \nstepexpectedcost<n><\strategy>.
	\end{equation*}
	By definition of $\Strategies'_n$, the witness sequence $\strategy_i \subseteq \Strategies_n$ has an accumulation point in $\Strategies_n$.
\end{proof}
\begin{proof}[Proof of \cref{stm:memoryless_after_var}]
	Let $\strategy^*$ be an optimal stationary policy minimizing the expected time to reach the goal states $\goalset$, which always exists \cite[Prop.~2]{DBLP:journals/mor/BertsekasT91}.
	Define $\strategy'$ as follows:
	For the first $n$ steps, copy $\strategy$, and starting in the $n$-th step, follow $\strategy^*$.
	Clearly, $\nstepprob<n><\strategy> = \nstepprob<n><\strategy'>$ and thus $\VaR_\threshold(\strategy) = \VaR_\threshold(\strategy')$ as well as $\nstepnongoal<n><\strategy> = \nstepnongoal<n><\strategy'>$.
	Additionally, we have $\nstepexpectedcost<n><\strategy> \geq \nstepexpectedcost<n><\strategy'>$.
	Together, we get $\CVaR_\threshold(\strategy) \geq \CVaR_\threshold(\strategy')$ by \cref{stm:cvar_equation}.
\end{proof}
\begin{proof}[Proof of \cref{stm:cvar_lp_solution}]
	\underline{First part}: Fix a policy $\strategy$ with $\VaR_\threshold(\strategy) = n$ and $\CVaR_\threshold(\strategy) = C$.
	We construct an assignment to the LP's variables:
	Set $p_{s, i} = \nstepprob<i><\strategy>(s)$ and $p_{s, a, i} = \Probability[s_i = s, a_i = a \mid \initialstate, \strategy]$.
	This assignment satisfies the first three constraints.
	For the fourth constraint, observe that by $\VaR_\threshold(\strategy) = n$, we have that $\sum_{s \in \goalset} \nstepprob<n - 1><\strategy>(s) < 1 - \threshold \leq \sum_{s \in \goalset} \nstepprob<n><\strategy>(s)$.
	By \cref{stm:cvar_equation}, we have that $\CVaR_\threshold(\strategy) = C = n + \tfrac{1}{\threshold} \nstepexpectedcost<n><\strategy>$.
	Since $\nstepexpectedcost<n><\strategy> = \sum_{s \in \States} \nstepprob<n><\strategy>(s) \cdot \expectedcost(s) = \sum_{s \in \States} p_{s, n} \cdot \expectedcost(s)$, we get that $\nstepexpectedcost<n><\strategy> = \threshold \cdot (c - n)$, proving the claim.

	\underline{Second part}: We construct the policy $\strategy$ as follows.
	For the first $n$ steps, at step $i$ in state $s$, choose action $a$ with probability $p_{s, a, i}$.
	Afterwards, i.e.\ starting from step $n$, in state $s$ follow a policy achieving the optimal expected cost $\expectedcost(s)$ (note the similarity to \cref{stm:memoryless_after_var}).
	Clearly, $\nstepprob<i><\strategy>(s) = p_{s, i}$ and thus $v = \nstepexpectedcost<n><\strategy>$.
	Now, we need to distinguish two cases.
	We have that $\VaR_\threshold(\strategy) = n - 1$ if $\sum_{s \in \goalset} p_{s, n - 1} = 1 - \threshold$ and $\VaR_\threshold(\strategy) = n$ otherwise.
	In the latter case, we directly get that $\CVaR_\threshold(\strategy) = n + \frac{1}{\threshold} v$ by \cref{stm:cvar_equation}.
	In the former, observe that in step $n - 1$ a fraction of exactly $1 - \threshold$ executions has reached the goal states.
	Consequently, the remaining $\nstepnongoal<n-1><\strategy> = \threshold$ executions all need to perform at least one more step, thus $\nstepexpectedcost<n - 1><\strategy> = \threshold + \nstepexpectedcost<n><\strategy>$.
	Inserting yields $\CVaR_\threshold(\strategy) = (n - 1) + \frac{1}{\threshold} \nstepexpectedcost<n - 1><\strategy> = n + \frac{1}{\threshold} \nstepexpectedcost<n><\strategy> = n + \frac{1}{\threshold} v$.
\end{proof}
\begin{proof}[Proof of \cref{stm:var_bound_from_cvar}]
	We have $\VaR_\threshold(\strategy^*) \leq \CVaR_\threshold(\strategy^*) \leq \CVaR_\threshold(\strategy)$, where the first inequality follows from definition and the second from optimality of $\strategy^*$.
\end{proof}
\begin{proof}[Proof sketch of \cref{stm:mdp_cvar_exponential_bound}]
	Since a proper policy exists, there also exists a proper memoryless deterministic policy $\strategy^p$ \cite[Prop.~2]{DBLP:journals/mor/BertsekasT91}.
	As $\strategy^p$ is deterministic, \cref{stm:var_mc_exponential} is applicable with $p_{\min}$ being the smallest transition probability in the MDP, and $\VaR_\threshold(\strategy^p)$ is at most exponential, as is $e(s)$ by similar reasoning.
	Together, $n + \tfrac{1}{\threshold} \nstepexpectedcost<n><\strategy^p>$ for $n = \VaR_\threshold(\strategy^p)$ is of at most exponential size, proving the claim through \cref{stm:cvar_equation}.
\end{proof}
\begin{proof}[Proof of \cref{stm:cvar_from_pareto}]
	\underline{First part}: Fix $n$ and $(p, E) \in \pareto<n><s>$ together with the policy $\strategy$ and set $\threshold = 1 - p$.
	Clearly, $\VaR_\threshold(\strategy) \leq n$ since at least a fraction of $p$ executions reach the goal within $n$ steps.
	Let $\VaR_\threshold(\strategy) = n' \leq n$ and $p' \leq p$ the exact fraction of executions that reach within $n'$ steps.
	By \cref{stm:cvar_equation} we get $\CVaR_\threshold(\strategy) = n' + \frac{1}{\threshold} \nstepexpectedcost<n'><\strategy>$.
	It remains to show that $n' + \frac{1}{\threshold} \nstepexpectedcost<n'><\strategy> \leq n + \frac{1}{\threshold} E$.
	After $n'$ steps, at most $1 - p'$ executions have not reached the goal.
	Thus $\nstepexpectedcost<n'><\strategy> \leq (1 - p') \cdot (n - n') + E \leq \threshold \cdot (n - n') + E$.

	\underline{Second part}: Fix a state $s$, a policy $\strategy$ and assume that $\VaR_\threshold(\strategy) = n$ and $\CVaR_\threshold(\strategy) = C$.
	We show that $(1 - \threshold, \threshold \cdot (C - n)) \in \pareto<n><s>$.
	The probability to reach the goal states in $n$ steps under $\strategy$ is at least $1 - \threshold$ by the definition of VaR, proving the first component.
	From \cref{stm:cvar_equation} we get $\CVaR_\threshold(\strategy) = C = n + \frac{1}{\threshold} \nstepexpectedcost<n><\strategy>$.
	Reordering yields $\nstepexpectedcost<n><\strategy> = \threshold \cdot (E - n)$, proving the second component.
\end{proof}
\begin{proof}[Proof of \cref{stm:pareto_shape}]
	\underline{Closure}:
		If we have that $(p, E) \in \pareto<n><s>$, we also have that $(p', E), (p, E') \in \pareto<n><s>$ for all $0 \leq p' \leq p$ and $E' \geq E$ by definition.

	\underline{Convexity}:
		Let $\strategy$ and $\strategy'$ be two policies corresponding to two points $(p, E), (p', E') \in \pareto<n><s>$ and fix $\lambda \in [0, 1]$.
		Following $\strategy$ with probability $\lambda$ and $\strategy'$ with probability $1 - \lambda$ reaches the goal set with at least $\lambda p + (1 - \lambda) p'$ and similar for the expectation. 

	\underline{Polygon}:
		We prove by induction that $\pareto<n><s>$ is a polygon with deterministic policies as vertices.

		For $n = 0$, we either have that $\pareto<0><s> = [0, 1] \times \Reals_{\geq 0}$ if $s \in \goalset$, or, if $s \notin \goalset$, $\pareto<0><s> = \{0\} \times [e(s), \infty)$.
		The first case follows trivially from the definition.
		For the second, observe that the probability to reach the goal in $0$ steps is zero and the expected time for all remaining executions is at least $e(s)$. 
		In both cases, $\pareto<0><s>$ is a polygon and the extremal points $(1, 0)$ and $(0, e(s))$, respectively, are achievable by a stationary deterministic policy.

		For the induction step, fix $n$ and a state $s$.
		We prove that $(p, E) \in \pareto<n+1><s>$ iff there exist a distribution over the actions $w : \stateactions(s) \to [0, 1]$ and achievable points $(p_{a, s'}, E_{a, s'}) \in \pareto<n><s'>$ for all $a \in \stateactions(s), s' \in \States$ such that
		\begin{equation*}
			(p, E) = {\sum}_{a \in \stateactions(s), s' \in \States} w(a) \cdot \mdptransitions(s, a, s') \cdot (p_{a, s'}, E_{a, s'})
		\end{equation*}
		The first equality follow directly from linearity of reachability:
		If the successors under action $a$ can reach the goal with probabilities $p_{a, s'}$ in $n$ steps, then the current state can reach the goal with the average of these probabilities in $n + 1$ steps (note the similarity to regular value iteration for reachability).
		For the second equality, recall the interpretation of the Pareto set:
		For all actions $a$ and successors $s'$ there exists a policy $\strategy_{a, s'}$ such that after $n$ steps at least a fraction of $p_{a, s'}$ executions have reached the goal and the expected time to reach the goal is at most $E_{a, s'}$.
		So, in state $s$, we can take one step and then simply follow these respective policies to achieve the values in the equation.
		Dually, if there is a policy $\strategy$ for $(p, E) \in \pareto<n+1><s>$, we immediately get policies achieving the respective values in the successors.
		Together with the induction hypothesis, this linear characterization proves that $\pareto<n+1><s>$ is a polygon, and the extremal points are achievable by Markovian deterministic policies.
\end{proof}
\begin{proof}[Proof of \cref{stm:cvar_algo_correct}]
	\underline{Termination}: We show that there always exists an $n$ such that $\pareto<n><\initialstate>$ contains $(1 - \threshold, E)$ for any $E$.
	Since we always have a proper policy $\strategy$, $\VaR_\threshold(\strategy) = n < \infty$ and thus $(1 - \threshold, \nstepexpectedcost<n><\strategy>) \in \pareto<n><\initialstate>$ by \cref{stm:cvar_from_pareto}.
	By \cref{stm:mdp_var_exponential_bound}, we know that $\VaR_\threshold(\strategy)$ is at most exponentially large, proving the step bound.

	\underline{Correctness}: Let $\strategy^*$ be an optimal policy, i.e.\ achieving the optimal CVaR.
	Further, let $\VaR_\threshold(\strategy^*) = n^*$ and $\CVaR_\threshold(\strategy^*) = E^*$.
	By \cref{stm:cvar_from_pareto}, we have that $(1 - \threshold, \threshold \cdot (E^* - n^*)) \in \pareto<n^*><\initialstate>$.
	Moreover, $\CVaR_\threshold(\strategy) \geq n^*$ for every policy $\strategy$, so the algorithm runs until at least iteration $n^*$.
	Consequently, $\mathtt{c} \leq E^*$ when the algorithm terminates.
	To conclude, if we had $\mathtt{c} < E^*$, there must exist another policy $\strategy'$ which achieves a better CVaR, again by virtue of \cref{stm:cvar_from_pareto}.
	Together, we have that $\mathtt{c} = E^*$ at the end.

	\underline{Runtime}:
	As argued above, the main loop is iterated at most exponentially often.
	Naively, we see that $\pareto<n><s>$ can have at most exponentially many vertices in $n$ by \cref{stm:pareto_shape}, as there are at most exponentially many deterministic policies.
	However, since we compute the convex hull of polygons, we can obtain a tighter bound.
	Recall that the Minkowski sum of two convex polygons with $i$ and $j$ vertices, respectively, has at most $i + j$ vertices.
	Thus, with $v_n(s)$ the number of vertices of $\pareto<n><s>$, we get $v_{n+1}(s) \leq \abs{\stateactions(s)} \cdot \sum_{s' \in \States} v_n(s')$.
	Consequently, $v_n(s) \in \mathcal{O}((\abs{\Actions} \abs{\States})^n)$ for all $s$, since $v_1(s) = 1$.
\end{proof}

\end{document}